\documentclass[article]{llncs}

\usepackage{amsmath,amssymb,mathtools}
\usepackage{booktabs}
\usepackage{tabularx}
\usepackage{algorithm}
\usepackage{microtype}
\usepackage[hidelinks]{hyperref}
\usepackage[nameinlink,capitalise,noabbrev]{cleveref}
\usepackage{tikz}
\newcommand{\EDC}{\textnormal{\textsc{Edge Deletion to Cactus}}}
\newcommand{\WEDC}{\textnormal{\textsc{Weighted Edge Deletion to Cactus}}}
\newcommand{\STC}{\textnormal{\textsc{Spanning Tree to Cactus}}}

\title{Exact Algorithms for Edge Deletion to Cactus Graphs and Weighted Variants}
\titlerunning{Exact Algorithms for Edge Deletion to Cactus Graphs and Weighted Variants}
\author{Wenhao Song}
\authorrunning{Wenhao Song}
\institute{
Department of Computer Science, Technion —– Israel Institute of Technology\\
\email{wenhao.song@campus.technion.ac.il}
}
\newenvironment{proof*}[1]
  {\renewcommand{\proofname}{#1}%
   \begin{proof}}
  {\end{proof}}

\begin{document}
\maketitle

\begin{abstract}
We study exact exponential-time algorithms for \EDC{}. Given a connected graph
$G$, the task is to delete a minimum number of edges so that the remaining
spanning graph is a connected cactus. Akhtar and Philip (IWOCA 2026) gave an
$O^*(3^n)$-time algorithm for the unweighted problem, where $n$ is the number of
vertices in the input graph and the $O^*()$ notation hides
polynomial factors.

We improve this algorithm to $O^*(2^n)$ time and space. More generally, if the
deletion costs take at most $q$ distinct nonnegative real values, then the
weighted problem can be solved in $O^*(2^n n^{O(q)})$ time and space. Thus
every fixed number of distinct costs, and in particular the unweighted case,
admits a faster exact algorithm. For nonnegative integer costs of total
weight $W$, we obtain an $O^*(2^n(W+1))$ pseudo-polynomial algorithm, while
arbitrary nonnegative real costs admit an $O^*(3^n)$ exact algorithm.
\end{abstract}

\section{Introduction}

Graph modification problems ask for a minimum number of local edits that
transform an input graph into a graph belonging to a prescribed class. In the
edge deletion version, the vertex set is fixed and one seeks a largest spanning
subgraph in the target class. Such problems go back at least to
Yannakakis~\cite{yannakakis1981} and El-Mallah and Colbourn~\cite{elmallahcolbourn1988}, 
see also Cai~\cite{cai1996graph}, Natanzon, Shamir, and Sharan~\cite{natanzonshamirsharan2001},
and Fomin and Kratsch~\cite{fominkratsch2010}.

We study the target class of \emph{cactus graphs}. A connected graph is a cactus
if every edge lies in at most one simple cycle. Equivalently, every block is
either a single edge or a simple cycle. Thus cactus graphs sit immediately
beyond trees: they allow local cycles while retaining a tree-like global
structure through their cut-vertex decomposition. They were studied classically
under the name Husimi trees~\cite{hararyuhlenbeck1953}.

This paper studies the following problem.

\begin{center}
\fbox{\begin{minipage}{0.92\textwidth}
\EDC{}

\smallskip
\noindent\textbf{Input:} A connected simple undirected graph $G=(V,E)$.

\smallskip
\noindent\textbf{Task:} Find a minimum cardinality set $F\subseteq E$ such that
$G-F$ is a connected cactus.
\end{minipage}}
\end{center}

Since only edges are deleted, every feasible solution is spanning. Hence
\EDC{} is equivalent to the problem of finding a connected spanning cactus
subgraph of $G$ with the maximum possible number of edges.

We also study the weighted version.

\begin{center}
\fbox{\begin{minipage}{0.92\textwidth}
\WEDC{}

\smallskip
\noindent\textbf{Input:} A connected simple undirected graph $G=(V,E)$ and a
nonnegative deletion cost function $c:E\to\mathbb{R}_{\ge0}$.

\smallskip
\noindent\textbf{Task:} Find a set $F\subseteq E$ minimizing
\[
        c(F)=\sum_{e\in F} c_e
\]
such that $G-F$ is a connected cactus.
\end{minipage}}
\end{center}

Equivalently, \WEDC{} asks for a connected spanning cactus $H\subseteq G$ that
maximizes $c(E(H))=\sum_{e\in E(H)}c_e$. The unweighted problem is the special
case in which every deletion cost is one.

\subsection{Prior work}

The problem is NP-hard on general graphs, as shown by El-Mallah and
Colbourn~\cite{elmallahcolbourn1988}. More recently, Koch, Pardal, and dos
Santos studied edge deletion to several tree-like graph classes and proved,
among other results, that deletion to cactus graphs remains hard even on bipartite input
graphs. They also gave positive results for chordal and quasi-threshold input
graphs~\cite{kochpardaldossantos2024}. These results place cactus deletion problem at
a natural boundary: deletion to trees is elementary, while deletion to the
slightly larger class of cactus graphs is already computationally hard.

The exact-exponential algorithm of the problem was studied recently by Akhtar
and Philip~\cite{akhtarphilip2026}. They considered \EDC{} together with the
related \STC{} problem. They gave a polynomial-time algorithm for the latter
using edge intersection graphs of paths in a tree~\cite{golumbicjamison1985b,golumbicjamison1985a},
which also yields an $O^*(n^{n-2})$ exact algorithm for \EDC{} by enumerating
spanning trees. They then gave a vertex-subset dynamic programming algorithm
with running time $O^*(3^n)$. Their algorithm is based on a cut vertex
split recurrence for maximum spanning cactus subgraphs. After choosing a vertex
$v$ as a possible cut vertex of a solution on $S$, the recurrence enumerates a
nonempty proper subset $A\subseteq S\setminus\{v\}$. The bottleneck is that
these split choices are evaluated separately for every vertex set $S$. Over all
sets $S$, this gives the $O^*(3^n)$ running time. They also identified weighted variants and related
cactus-like target classes as natural directions.

\subsection{Our results}

Our first result is an exact algorithm for the weighted problem when the
number of distinct deletion costs is bounded.

\begin{theorem}[Few distinct real weights]\label{thm:fewweights-main}
Let $G=(V,E)$ be a connected graph on $n$ vertices, and suppose that the deletion
costs take at most $q$ distinct nonnegative real values. Then \WEDC{} can be
solved exactly in
\[
        O^*(2^n n^{O(q)})
\]
time and space in the comparison-addition model. In particular, for every fixed
$q$, the running time is $O^*(2^n)$.
\end{theorem}

The faster algorithm for the unweighted problem is obtained by assigning deletion cost $1$ to every edge,
so it is exactly the $q=1$ case of \cref{thm:fewweights-main}.

\begin{corollary}[Unweighted case]\label{cor:unweighted-main}
\EDC{} can be solved in $O^*(2^n)$ time and $O^*(2^n)$ space.
\end{corollary}

For integer weights we obtain a pseudo-polynomial version of the same
convolution acceleration. This follows from the same dynamic programming algorithm, using the bounded max-sum
subset convolution of Lemma~\ref{lem:maxsum}. The proof is deferred to
Appendix~\ref{app:integer}.

\begin{theorem}\label{thm:int-main}
Let all deletion costs be nonnegative integers, and let $W=\sum_{e\in E}c_e$.
Then the weighted problem can be solved in pseudo-polynomial time and space
$O^*(2^n(W+1))$.
\end{theorem}

For arbitrary nonnegative real deletion costs, where the number of distinct
weights may be unbounded, we retain the direct split evaluation and obtain the
following weighted exact algorithm.

\begin{theorem}\label{thm:weighted-main}
Let $G=(V,E)$ be a connected graph on $n$ vertices, and let
$c:E\to\mathbb{R}_{\ge 0}$ be nonnegative deletion costs. In the
comparison-addition model, \WEDC{} can be solved in $O^*(3^n)$ time and
$O^*(2^n)$ space.
\end{theorem}

Here the comparison-addition model is the standard exact real-arithmetic model
used for weighted graph algorithms: stored weight values may be copied, added,
subtracted, and compared exactly in unit time. We do not use multiplication,
division, floor, modulo, or bit operations on arbitrary real weights. In the
few-distinct-weights algorithm, comparing two count vectors
$\mu,\nu\in\{0,\ldots,n\}^q$ means comparing the real values
$\sum_i\mu_i\alpha_i$ and $\sum_i\nu_i\alpha_i$, the resulting polynomial
overhead is included in the factor $n^{O(q)}$. For integer weights, the
dependence on $W+1$ is displayed explicitly, with polylogarithmic factors in
$W+1$ suppressed.

\begin{table}[t]
\centering
\begin{tabularx}{\textwidth}{@{}lX@{}}
\toprule
Variant & Running time and space\\
\midrule
Unweighted deletion costs & $O^*(2^n)$ time and space \\
At most $q$ distinct real deletion costs & $O^*(2^n n^{O(q)})$ time and space \\
Nonnegative integer deletion costs, total weight $W$ & $O^*(2^n(W+1))$ pseudo-polynomial time and space \\
Arbitrary nonnegative real deletion costs & $O^*(3^n)$ time, $O^*(2^n)$ space \\
\bottomrule
\end{tabularx}
\caption{Main algorithmic results.}\label{tab:results}
\end{table}

Our results are summarized in~\cref{tab:results}.
\subsection{Technical overview}

The central observation is that a cactus is a tree of blocks. Therefore the
dynamic programming algorithm has two tasks: compute the best possible block on
each vertex set, and combine smaller solutions through cut vertices. Our
starting point is the cut vertex recurrence of Akhtar and Philip~\cite{akhtarphilip2026}. 
However, our recurrence is slightly different. The unweighted argument relies
on the fact that a maximum spanning cactus that has
a cut vertex always exists. This does not extend directly to the weighted setting: 
every maximum-weight spanning cactus may consist of a single block (as shown in~\cref{ex:one-block-necessity}). 
Thus we retain the cut vertex recurrence, but explicitly include the one-block
term and change how the split term is evaluated. Instead of trying every
set $A\subseteq S\setminus\{v\}$ independently for every $S$, we use \emph{fast subset convolution} 
that was introduced in~\cite{bjorklund2007} to compute the split values simultaneously. This is what
removes the $O(3^n)$ split enumeration in the unweighted case. The same approach
extends to weighted instances whenever the objective values have a compact
discrete encoding.

For a nonempty vertex subset $S\subseteq V$, let $D(S)$ be the maximum weight of a
connected cactus spanning $S$ inside $G[S]$. Such a cactus has one of two
forms. Either it has no cut vertex, in which case it is a single block, or it
has a cut vertex $v$. In the second case, deleting $v$ separates the cactus
into components. Grouping these components into two nonempty subsets gives two
smaller cactus graphs, both containing $v$.

This gives the recurrence
\[
D(S)=\max\left( B(S),
\max_{\substack{v\in S\\ \emptyset\ne A\subsetneq S\setminus\{v\}}}
\left[D(A\cup\{v\})+D(S\setminus A)\right]\right),
\]
where $B(S)$ is the best one-block cactus weight on $S$. By the block characterization of cactus graphs, a one-block cactus is either the one-vertex
graph, a single edge, or a Hamiltonian cycle on $S$. Hence all values $B(S)$ can be
computed in $O^*(2^n)$ time by the Bellman-Held-Karp dynamic programming
algorithm~\cite{bellman1962,heldkarp1962}.

After the one-block table is known, direct evaluation of the split term costs
$O^*(3^n)$: for every set $S$, one enumerates a cut vertex $v$ and a subset
$A\subseteq S\setminus\{v\}$. If we fix $v$ and view the two subsets of the split as a
disjoint decomposition of $S\setminus\{v\}$, the resulting maximization is a
max-sum subset convolution. In the unweighted case, the value being combined is
just the number of retained edges, so it has only $O(n)$ possible values and can
be encoded by a single polynomial variable. This is exactly the case $q=1$.
If the edge costs belong to $\{\alpha_1,\ldots,\alpha_q\}$, the integer count is
replaced by a count vector $\mu$, where $\mu_i$ counts the number of retained edges of cost
$\alpha_i$. The scalar value is $\alpha\cdot\mu$. Since a cactus has only
$O(n)$ edges, there are only $n^{O(q)}$ possible count vectors, so ordinary
subset convolution over a truncated multivariate polynomial ring computes all
split values with only an $n^{O(q)}$ overhead. For integer weights, one uses a
single variable recording the total retained weight, giving degree at most $W+1$.

The distinction between few distinct weights and arbitrary real weights is
algorithmically important. Fast max-sum subset convolution gives the $O^*(2^n)$
running time only when the values being maximized have a compact discrete
encoding. This is automatic in the unweighted case, true for $q$ distinct
weights through count vectors, and true for integer weights through total
weight. For arbitrary real weights, no such compact encoding is available in
this framework, so we keep the $O^*(3^n)$ split evaluation.

\subsection{Organization}

Our paper is organized as follows. \Cref{sec:preliminaries} gives notation, recalls the cactus block
characterization and introduces subset convolution. \Cref{sec:weighted} proves the weighted
$O^*(3^n)$ algorithm. \Cref{sec:fewweights} proves \Cref{thm:fewweights-main} and derives the faster unweighted algorithm as a
consequence. \cref{sec:discussion} concludes.

\section{Preliminaries}\label{sec:preliminaries}

\subsection{Graph notation and cactus graphs}

Let $G=(V,E)$ be a graph. For $S\subseteq V$, $G[S]$ denotes the subgraph induced by $S$. For a subgraph $H\subseteq G$, $V(H)$ and $E(H)$ denote its vertex and edge sets respectively. A subgraph $H\subseteq G[S]$ is said to \emph{span} $S$ if $V(H)=S$. A \emph{cut vertex} of a connected graph is a vertex whose deletion disconnects the graph. A \emph{block} is a maximal connected subgraph with no cut vertex. 
Equivalently, in standard graph-theoretic terminology, a block is a maximal biconnected component. Blocks can be found in linear time by the classical Hopcroft-Tarjan biconnected-component algorithm~\cite{hopcrofttarjan1973}.

We use the following standard and well-known characterization of cactus graphs (see, e.g.,
White~\cite[Ch.~6, p.~57]{white2001graphs}).

\begin{lemma}[Characterization of cactus graph]\label{lem:block-char}
For a connected graph $H$, the following conditions are equivalent.
\begin{enumerate}
    \item $H$ is a cactus.
    \item Any two simple cycles of $H$ share at most one vertex.
    \item Every block of $H$ is either a single edge or a simple cycle.
\end{enumerate}
\end{lemma}

The one-vertex graph is treated as a degenerate connected cactus. This convention avoids exceptional cases in our dynamic programming algorithms.

\subsection{Deletion costs and subgraph weights}

For an edge cost function $c:E\to\mathbb{R}_{\ge 0}$ and an edge set
$X\subseteq E$, write
\[
        c(X)=\sum_{e\in X}c_e.
\]
For a subgraph $H\subseteq G$, its weight is $c(E(H))$. We use $-\infty$ to
mark infeasible dynamic programming states.

The weighted problem was stated in the introduction as a minimization problem
over the deleted edge set. Because only edges are deleted, every feasible
solution has the form $G-F$ and spans the same vertex set as the input graph.
Thus minimizing the cost of deleted edges is equivalent to maximizing the weight
of the resulting spanning cactus subgraph.

\begin{lemma}\label{lem:delete-retain}
Let $C_{\mathrm{tot}}=c(E)$. If $H=G-F$, then
\[
        c(F)=C_{\mathrm{tot}}-c(E(H)).
\]
Consequently, minimizing the deleted cost is equivalent to maximizing $c(E(H))$
over all connected spanning cactus subgraphs $H\subseteq G$.
\end{lemma}

\subsection{Subset convolution for split maximization}

Let $N$ be a finite set. The cut vertex recurrence contains a split term of the form
\[
        \max_{Y\subseteq X} \bigl(f(Y)+g(X\setminus Y)\bigr).
\]
If this maximum is evaluated separately for every $X\subseteq N$, then the total
number of pairs $(X,Y)$ with $Y\subseteq X$ is $3^{|N|}$. The purpose of subset
convolution is to compute all these split values simultaneously in
$O^*(2^{|N|})$ time, provided the values can be encoded in a suitable polynomial ring.

We recall the ordinary subset convolution of Bj\"orklund, Husfeldt, Kaski, and
Koivisto~\cite{bjorklund2007}. Let $R$ be a ring, and let $\mathcal{P}(N)$ denote the power set of $N$. For functions $f,g:\mathcal{P}(N)\to R$, their subset convolution is
\[
        (f*g)(X)=\sum_{Y\subseteq X}f(Y)g(X\setminus Y).
\]
Direct evaluation for all $X\subseteq N$ costs $\Theta(3^{|N|})$ ring operations.
Fast subset convolution computes all values in $O(|N|^2 2^{|N|})$ ring operations
by ranked zeta transforms and M\"obius inversion~\cite{bjorklund2007}.

Our recurrence needs a max-sum version. We use a 
standard encoding trick: an integer value $a$ is represented by a monomial
$y^a$, and a vector value $\mu$ is represented by a multivariate monomial
$z^\mu$. After ordinary subset convolution is performed on these monomials, the
largest exponent, or the best exponent vector under the objective, gives the
maximum split value.

The next lemma is the bounded integer max-sum subset convolution theorem of
Bj\"orklund, Husfeldt, Kaski, and Koivisto~\cite[Theorem~3]{bjorklund2007},
specialized to nonnegative ranges.
\begin{lemma}[Bounded max-sum subset convolution]\label{lem:maxsum}
Let $f,g:\mathcal{P}(N)\to\{0,1,\ldots,M\}\cup\{-\infty\}$. All values of
\[
        (f\star g)(X)=\max_{Y\subseteq X}
        \bigl(f(Y)+g(X\setminus Y)\bigr)
\]
can be computed in $O^*(2^{|N|}(M+1))$ time and
$O^*(2^{|N|}(M+1))$ space, up to polylogarithmic factors in $M+1$.
\end{lemma}

The following lemma is a vector-valued extension of~\Cref{lem:maxsum}, we state it in the count vector form needed below. A proof is included in Appendix~\ref{app:count-vector-conv}
\begin{lemma}[Count vector subset convolution]\label{lem:count-vector-conv}
Let $N$ be a finite ground set, let $\alpha=(\alpha_1,\ldots,\alpha_q)\in\mathbb{R}_{\ge0}^q$, where $\mathbb{R}_{\ge0}$ denotes the set of nonnegative real numbers, and let $L$ be a nonnegative integer. Suppose each feasible value of two functions $p,r:\mathcal{P}(N)\to\{0,1,\ldots,L\}^q\cup\{\bot\}$ is a count vector, while $\bot$ denotes infeasibility. For every $X\subseteq N$, define
\[
        h(X)=\max_{Y\subseteq X}\alpha\cdot\bigl(p(Y)+r(X\setminus Y)\bigr),
\]
where splits with $p(Y)=\bot$ or $r(X\setminus Y)=\bot$ are ignored, and $\alpha\cdot\mu=\sum_i\alpha_i\mu_i$. Then all values $h(X)$, together with an attaining count vector, can be computed in
\[
        O^*(2^{|N|} L^{O(q)})
\]
time and space in the comparison-addition model.
\end{lemma}

\section{The weighted dynamic programming algorithm}\label{sec:weighted}

This section proves \cref{thm:weighted-main}. The algorithm has two ingredients. First, it computes the best possible one-block cactus on every vertex subset. Second, it combines smaller maximum cactus subgraphs by gluing them at cut vertices, which was the original method in~\cite{akhtarphilip2026}.

\subsection{One-block values}

For a nonempty set $S\subseteq V$, define
\[
\begin{aligned}
B(S)=\max\{c(E(H)) :{}& H\subseteq G[S],\ V(H)=S,\\
                     & H\text{ is a connected cactus with no cut vertex}\}.
\end{aligned}
\]
If no such graph exists, set $B(S)=-\infty$.

The block cacti are characterized.

\begin{lemma}[Characterization of a one-block cactus]\label{lem:one-block}
Let $H$ be a connected cactus on vertex set $S$ with no cut vertex. Then exactly one of the following holds.
\begin{enumerate}
    \item $|S|=1$ and $H$ is the one-vertex graph.
    \item $|S|=2$ and $H$ consists of a single edge.
    \item $|S|\ge 3$ and $H$ is a simple cycle containing all vertices of $S$.
\end{enumerate}
\end{lemma}

\begin{proof}
If $|S|=1$, the statement is immediate. If $|S|\ge2$ and $H$ has no cut vertex, then $H$ consists of a single block. By Lemma~\ref{lem:block-char}, a block of a cactus is either a single edge or a simple cycle. A single edge has exactly two vertices. In all remaining cases the unique block is a simple cycle using all vertices of $S$.
\end{proof}

Consequently,
\[
B(S)=
\begin{cases}
0, & |S|=1,\\
c_{uv}, & S=\{u,v\}\text{ and }uv\in E,\\
\mathrm{HC}(S), & |S|\ge3,\\
-\infty, & \text{otherwise,}
\end{cases}
\]
where $\mathrm{HC}(S)$ denotes the maximum weight of a Hamiltonian cycle in
$G[S]$. 

\begin{lemma}\label{lem:B-computation}
All values $B(S)$ can be computed in $O^*(2^n)$ time and $O^*(2^n)$ space.
\end{lemma}

\begin{proof}
        The cases $|S|=1$ and $|S|=2$ are handled directly. If $|S|\ge 3$, use the Bellman-Held-Karp dynamic programming algorithm~\cite{bellman1962,heldkarp1962}.
\end{proof}

\subsection{The weighted recurrence}

For every nonempty $S\subseteq V$, define
\[
        D(S)=\max\{c(E(H)): H\subseteq G[S],\ V(H)=S,\ H\text{ is a connected cactus}\}.
\]
If no connected cactus spanning $S$ exists, set $D(S)=-\infty$. Since the original graph $G$ is connected, $D(V)$ is finite.

The cut vertex part of the following recurrence is the weighted analogue of the
unweighted decomposition used by Akhtar and Philip~\cite[Theorems~6 and~7]{akhtarphilip2026}. We give a self-contained proof
because our formulation keeps the one-block value explicitly and works with
arbitrary nonnegative deletion costs.

We claim that $D(S)$ is given by
\begin{equation}\label{eq:weighted-rec}
D(S)=
\max\left(
B(S),
\max_{\substack{v\in S\\ \emptyset\ne A\subsetneq S\setminus\{v\}}}
\left[
D(A\cup\{v\})+D(S\setminus A)
\right]
\right).
\end{equation}

The following lemma is the structural observation of Akhtar and Philip~\cite[Lemma~4]{akhtarphilip2026}. It is critical to designing the algorithm.

\begin{lemma}\label{lem:gluing-cacti}
Let $H_1$ and $H_2$ be connected cactus graphs such that
$V(H_1)\cap V(H_2)=\{v\}$ and $E(H_1)\cap E(H_2)=\emptyset$.
Then $H_1\cup H_2$ is a connected cactus.
\end{lemma}

\begin{lemma}[Correctness of the recurrence]\label{lem:recurrence-correct}
For every nonempty $S\subseteq V$, \cref{eq:weighted-rec} gives the value $D(S)$.
\end{lemma}

\begin{proof}
We prove the statement by induction on $|S|$. For $|S|<3$, it is obviously correct. We assume now that $|S|\ge 3$.

Every value on the right-hand side is feasible. The value $B(S)$ is feasible by definition. For a split term, take maximum cactus subgraph $H_1$ on $A\cup\{v\}$ and $H_2$ on $S\setminus A$. Their vertex sets meet exactly in $v$, and their edge sets are disjoint because the two subsets use disjoint vertex sets apart from $v$. By Lemma~\ref{lem:gluing-cacti}, $H_1\cup H_2$ is a connected cactus spanning $S$, with weight equal to the sum of the two weights.

Conversely, let $H$ be a maximum connected cactus spanning $S$. If $H$ has no cut vertex, then its weight is at most $B(S)$, so it is captured by the first term. Otherwise let $v$ be a cut vertex of $H$. The components of $H-v$ partition $S\setminus\{v\}$. Choose a nonempty proper union $A$ of these components, and put $C=(S\setminus\{v\})\setminus A$. Then $H[A\cup\{v\}]$ and $H[C\cup\{v\}]$ are connected cactus subgraphs. By induction, their weights are at most $D(A\cup\{v\})$ and $D(C\cup\{v\})$. Since their edge sets partition $E(H)$, the weight of $H$ is at most the corresponding split value. This proves the reverse inequality.
\end{proof}

\begin{algorithm}[H]
\caption{\textsc{Weighted-Cactus-DP}}
\label{alg:weighted}
\begin{enumerate}
    \item Compute all one-block values $B(S)$ by the Bellman-Held-Karp recurrence.
    \item For $k=1,2,\ldots,n$ and for every $S\subseteq V$ with $|S|=k$:
    \begin{enumerate}
        \item Set $D(S)\leftarrow B(S)$.
        \item For every $v\in S$ and every nonempty proper subset $A\subsetneq S\setminus\{v\}$, update
        \[
        D(S) \leftarrow  \max\{D(S),D(A\cup\{v\})+D(S\setminus A)\}.
        \]
    \end{enumerate}
    \item Output $\sum_{e\in E}c_e-D(V)$.
\end{enumerate}
\end{algorithm}

\begin{proof*}{Proof of \cref{thm:weighted-main}}
We present~\Cref{alg:weighted} for the problem. The recurrence is evaluated in increasing order of $|S|$. In every split, both recursive sets have size smaller than $|S|$.

By~\Cref{lem:recurrence-correct}, \Cref{alg:weighted} computes all values $D(S)$
correctly, and in particular computes $D(V)$. By Lemma~\ref{lem:B-computation}, the table $B(S)$ is computed in
$O^*(2^n)$ time and space. The number of split checks is at most
\[
        \sum_{S\subseteq V}|S|2^{|S|-1}=n3^{n-1}.
\]
Thus the total time is $O^*(3^n)$, and the space is $O^*(2^n)$.
Finally, Lemma~\ref{lem:delete-retain} converts the maximum retained weight
$D(V)$ into the minimum deletion cost.

To reconstruct an optimal deletion set, store for each $S$ whether the maximum
came from $B(S)$ or from a split $(v,A)$, and recursively reconstruct the
chosen cactus.
\end{proof*}

\section{Few distinct real weights}\label{sec:fewweights}

This section proves \cref{thm:fewweights-main}. The structural recurrence is still \cref{eq:weighted-rec}. The new point is that the split term can be accelerated when the edge costs take only few distinct values using~\Cref{lem:count-vector-conv}.

Assume the set of edge costs is contained in
\[
        \{\alpha_1,\ldots,\alpha_q\}\subseteq\mathbb{R}_{\ge0}.
\]
For a cactus graph $H$, define its \emph{count vector}
\[
        \mu(H)=(\mu_1(H),\ldots,\mu_q(H)),
\]
where $\mu_i(H)$ is the number of edges of cost $\alpha_i$. The weight is
\[
        c(E(H))=\alpha\cdot\mu(H)=\sum_{i=1}^q \alpha_i\mu_i(H).
\]
It was shown in~\cite[p.~160]{west2001} that every simple cactus on $s$ vertices has at most $\lfloor \frac{3(s-1)}{2} \rfloor$ edges. Hence each coordinate of every count vector that occurs in the dynamic programming algorithm is at most $L=\frac{3n}{2}$. The number of possible count vectors is therefore $n^{O(q)}$.

The dynamic programming algorithm stores, for every set $S$, both the scalar optimum $D(S)$ and one count vector $\mu(S)$ of a maximum cactus attaining it. If no cactus spanning $S$ exists, we store infeasibility. 
If several candidate solutions have the same maximum value, the algorithm may choose and store any one of them. Although a set $S$ may admit several maximum cactus subgraphs with different count vectors, one attaining vector is enough: later recurrences use only the scalar optimum $D(S)$. Thus all count vectors with the same scalar value are interchangeable for the purpose of computing the optimum.

\subsection{One-block values with count vectors}

The one-block values are computed as in \cref{sec:weighted}. The only change is that each Bellman-Held-Karp state stores a count vector in addition to the scalar value. When an edge of cost $\alpha_i$ is added, the $i$th coordinate of the count vector is increased by one. Comparisons are made using the real scalar value $\alpha\cdot\mu$. This preserves the $O^*(2^n n^{O(q)})$ time complexity. In fact the one-block table still has only $O^*(2^n)$ states, with a polynomial overhead for storing and comparing count vectors.

Let $\beta(S)$ denote a count vector attaining the one-block value $B(S)$, when $B(S)$ is feasible.

\subsection{Layered split computation}

As before, the algorithm computes the values by increasing cardinality. Suppose all values for sets of size smaller than $k$ have already been computed. Fix a candidate cut vertex $v\in V$. Define a function on subsets of $V\setminus\{v\}$ by
\[
        p_{v,k}(A)=
        \begin{cases}
        \mu(A\cup\{v\}), & 1\le |A|\le k-2\text{ and }D(A\cup\{v\})\ne -\infty,\\
        \bot, & \text{otherwise.}
        \end{cases}
\]
Apply Lemma~\ref{lem:count-vector-conv} to $p_{v,k}$ with itself. For every $X\subseteq V\setminus\{v\}$, the result is
\[
        \mathrm{Split}_{v,k}(X)=\max_{A\subseteq X}\alpha\cdot\bigl(\mu(A\cup\{v\})+\mu((X\setminus A)\cup\{v\})\bigr),
\]
ignoring invalid subsets. When $S$ has size $k$, contains $v$, and $X=S\setminus\{v\}$, this is exactly the best split value of $S$ using $v$ as the cut vertex.

The recurrence for a $k$-vertex set $S$ is therefore
\begin{equation}        \label{eq:fewweights}
        D(S)=\max\left(B(S),\max_{v\in S}\mathrm{Split}_{v,k}(S\setminus\{v\})\right),
\end{equation}
Along with $D(S)$, the algorithm stores the corresponding count vector: either $\beta(S)$ from the one-block case, or the sum of the two count vectors returned by the split convolution.

\begin{algorithm}[H]
\caption{\textsc{Few-Weights-Cactus-DP}}
\label{alg:fewweights}
\begin{enumerate}
    \item Compute all one-block values $B(S)$ and corresponding count vectors $\beta(S)$.
    \item For $k=1,2,\ldots,n$:
    \begin{enumerate}
        \item Initialize $D(S)\leftarrow B(S)$ and $\mu(S)\leftarrow\beta(S)$ for all $S$ with $|S|=k$.
        \item For each $v\in V$:
        \begin{enumerate}
            \item Form $p_{v,k}$ from already computed smaller sets.
            \item Apply Lemma~\ref{lem:count-vector-conv} once to $p_{v,k}$
            with itself, thereby computing
            \[
            \operatorname{Split}_{v,k}(X)
            \]
            for all $X\subseteq V\setminus\{v\}$ simultaneously.
            \item For every $S$ with $|S|=k$ and $v\in S$, update $D(S)$ and $\mu(S)$ if the split value for $S\setminus\{v\}$ is better.
        \end{enumerate}
    \end{enumerate}
    \item Output $\sum_{e\in E}c_e-D(V)$.
\end{enumerate}
\end{algorithm}

\begin{lemma}\label{lem:fewweights-layer}
\cref{eq:fewweights} is the correct recurrence for $D(S)$.
\end{lemma}

\begin{proof}
The scalar recurrence is exactly \cref{eq:weighted-rec}. The one-block term is represented by $B(S)$ and $\beta(S)$. For a cut vertex split at $v$, both subsets have smaller cardinality, so by the induction hypothesis their stored count vectors attain their scalar optima. Since weight is additive under gluing, the scalar value of the split is the dot product of $\alpha$ with the sum of the two stored vectors. Lemma~\ref{lem:count-vector-conv} maximizes this quantity over all choices of the subset $A\subseteq S\setminus\{v\}$ simultaneously. Taking the maximum over $v$ and comparing with the one-block value gives exactly \cref{eq:weighted-rec}. The stored vector is a witness for the chosen optimum.
\end{proof}

\begin{proof*}{Proof of \cref{thm:fewweights-main}}
By~\cref{lem:fewweights-layer}, \cref{alg:fewweights}~correctly computes $D(V)$.

The count vector dimension is $q$, and every coordinate is bounded by $L=\frac{3n}{2}$. For every layer $k$ and vertex $v$, Lemma~\ref{lem:count-vector-conv} performs one subset convolution on the ground set $V\setminus\{v\}$ in
\[
        O^*(2^n L^{O(q)})=O^*(2^n n^{O(q)})
\]
time and space. There are only polynomially many pairs $(k,v)$. The one-block table is computed within the same complexity. Thus the total running time and space are $O^*(2^n n^{O(q)})$. Finally, Lemma~\ref{lem:delete-retain} converts the maximum weight to the minimum deletion cost.

To reconstruct a maximum cactus, store for each set $S$ whether the optimum comes from $B(S)$ or from a split vertex $v$. For a chosen split, an attaining subset $A\subseteq S\setminus\{v\}$ can be recovered by scanning subsets $A$ and checking the equality
\[
D(S)=D(A\cup\{v\})+D(S\setminus A).
\]
Along the single recursion tree this adds only $O^*(2^n)$ time.
\end{proof*}

\section{Conclusion}\label{sec:discussion}

We gave faster exact algorithms for \EDC{} and weighted
variants. The main result is the $O^*(2^n)$ algorithm for the unweighted
problem, improving the previous $O^*(3^n)$ exact dynamic program of Akhtar and
Philip~\cite{akhtarphilip2026}. This is the $q=1$ case of the same
count vector method that gives an $O^*(2^n n^{O(q)})$ algorithm when the
deletion costs take at most $q$ distinct real values. For nonnegative integer
weights with total weight $W$, we obtain an $O^*(2^n(W+1))$ pseudo-polynomial
algorithm. For arbitrary nonnegative real weights, we obtain an $O^*(3^n)$
algorithm by explicitly keeping the one-block table and evaluating cut vertex
splits directly.

All algorithms here are vertex-subset dynamic programming algorithms and therefore use exponential space. It would be interesting to obtain
a polynomial-space exact algorithm with comparable running time.

\appendix

\section{Integer weights}\label{app:integer}

Assume now that $c_e\in\mathbb{Z}_{\ge0}$ for all edges and let
\[
        W=\sum_{e\in E}c_e.
\]
All dynamic programming values lie in
$\{0,1,\ldots,W\}\cup\{-\infty\}$. Thus
Lemma~\ref{lem:maxsum} applies with $M=W$.

The one-block values are maximum-weight Hamiltonian-cycle values, still
bounded by $W$. They can be computed by the same Bellman-Held-Karp
algorithm. The split term is evaluated layer by layer using the one-variable
bounded max-sum convolution of Lemma~\ref{lem:maxsum}.

\begin{proof*}{Proof of Theorem~\ref{thm:int-main}}
The proof is almost the same as that of~\cref{thm:fewweights-main}. 
Use the same dynamic programming algorithm as~\cref{alg:fewweights}, but encode each
scalar weight by a one-variable monomial. For every layer $k$ and vertex
$v$, Lemma~\ref{lem:maxsum} computes the required max-sum convolution in
$O^*(2^n(W+1))$ time, up to polylogarithmic factors in $W+1$. The number
of layers and vertices is polynomial. The space complexity is the same order.
Finally, Lemma~\ref{lem:delete-retain} converts the maximum weight to the
minimum deletion cost.
\end{proof*}

This is a pseudo-polynomial algorithm. If weights are encoded in binary,
$W$ may be exponential in the input length.

\section{Proof of~\cref{lem:count-vector-conv}} \label{app:count-vector-conv}

\begin{proof*}{Proof of~\cref{lem:count-vector-conv}}
We reduce the statement to ordinary fast subset convolution over rings.

Work over the truncated polynomial ring
\[
        R=\mathbb{Z}[z_1,\ldots,z_q]/
        (z_1^{2L+1},\ldots,z_q^{2L+1}).
\]
For a feasible vector $\mu=(\mu_1,\ldots,\mu_q)$, write
\[
        z^\mu=z_1^{\mu_1}\cdots z_q^{\mu_q}.
\]
Encode
\[
        \widetilde p(A)=
        \begin{cases}
        z^{p(A)}, & p(A)\ne\bot,\\
        0, & p(A)=\bot,
        \end{cases}
        \qquad
        \widetilde r(A)=
        \begin{cases}
        z^{r(A)}, & r(A)\ne\bot,\\
        0, & r(A)=\bot.
        \end{cases}
\]
Compute the ordinary subset convolution
\[
        H=\widetilde p*\widetilde r
\]
over $R$. Thus, for every $X\subseteq N$,
\[
        H(X)=\sum_{Y\subseteq X}\widetilde p(Y)\widetilde r(X\setminus Y).
\]

Since every coordinate of $p(Y)$ and $r(X\setminus Y)$ is at most $L$,
every feasible exponent vector $p(Y)+r(X\setminus Y)$ has coordinates at
most $2L$. Hence no feasible monomial is removed by the truncation.
Therefore a monomial $z^\nu$ has nonzero coefficient in $H(X)$ if and
only if there is a feasible split $Y\subseteq X$ such that
\[
        \nu=p(Y)+r(X\setminus Y).
\]
In the final polynomial $H(X)$, all coefficients are nonnegative integers,
so there is no cancellation. Scanning all exponent vectors with nonzero
coefficient and choosing one maximizing $\alpha\cdot\nu$ gives $h(X)$
and the corresponding count vector whenever $h(X)>-\infty$. If no such
monomial exists, we return $h(X)=-\infty$ and no count vector.

The ring $R$ has $(2L+1)^q=L^{O(q)}$ monomials. Thus each ring
operation can be implemented in $L^{O(q)}$ integer-arithmetic time
and space. The integer coefficients that arise have polynomially many bits
in the input parameters, and the real numbers $\alpha_i$ are used only
when comparing values $\alpha\cdot\nu$. Such comparisons use only exact
additions and comparisons of the real weights, with an additional
$L^{O(q)}$ overhead. Hence the total time and space are $O^*(2^{|N|}L^{O(q)})$.

\end{proof*}

\section{An example showing the necessity of the one-block term}
\label{app:weighted-example}
The following example explains why the one-block table is needed in the weighted
setting.

\begin{example}\label{ex:one-block-necessity}
        Let $G$ be the graph shown in~\cref{fig:no-cut-vertex}. $G$ consists of the $5$-cycle $C=v_1v_2v_3v_4v_5v_1$ together with the chord $v_1v_3$. Assign weight $2$ to every edge of
        $C$, and assign weight $1$ to the chord $v_1v_3$. 
        
        In the unweighted setting, the maximum size of a spanning cactus subgraph is
        $5$. Besides the cycle $C$, there is a maximum-sized spanning cactus with a cut vertex. For example, $H=G-v_1v_2$ 
        consists of the cycle $v_1v_3v_4v_5v_1$ and the edge $v_2v_3$. Hence $v_3$ is a cut vertex of $H$, and $|E(H)|=|E(C)|=5$. 
        
        In the weighted setting, however, $C$ has weight $c(E(C))=5\cdot 2=10$. 
        Any spanning cactus containing the chord $v_1v_3$ must omit at least one
        edge of $C$, since $G$ itself is not a cactus. Its weight is therefore at
        most $4\cdot 2+1=9$. Consequently, $C$ is the unique maximum-weight spanning cactus of $G$.
        In particular, every maximum-weight spanning cactus has no cut vertex.
\end{example}

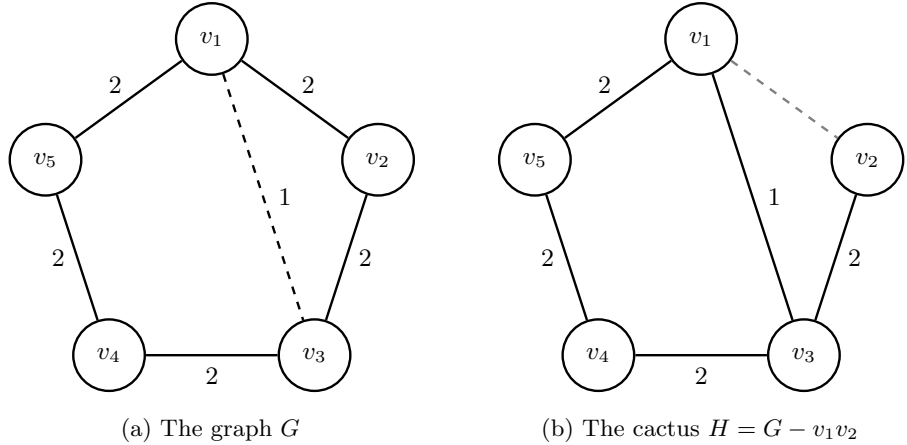
\begin{figure}[h]
        \centering
                
        \begin{minipage}{0.47\textwidth}
        \centering
        \begin{tikzpicture}[
                scale=1.05,
                transform shape,
                vertex/.style={
                circle,
                draw,
                fill=white,
                inner sep=1.5pt,
                minimum size=9mm
                },
                weight/.style={
                draw=none,
                fill=none,
                inner sep=0pt
                },
                every path/.style={line width=0.9pt}
        ]
                \node[vertex] (v1) at (90:2.2)   {$v_1$};
                \node[vertex] (v2) at (18:2.2)   {$v_2$};
                \node[vertex] (v3) at (-54:2.2)  {$v_3$};
                \node[vertex] (v4) at (-126:2.2) {$v_4$};
                \node[vertex] (v5) at (162:2.2)  {$v_5$};
                
                \draw (v1) -- node[weight, midway, above right=2pt] {$2$} (v2);
                \draw (v2) -- node[weight, midway, right=4pt] {$2$} (v3);
                \draw (v3) -- node[weight, midway, below=4pt] {$2$} (v4);
                \draw (v4) -- node[weight, midway, left=4pt] {$2$} (v5);
                \draw (v5) -- node[weight, midway, above left=2pt] {$2$} (v1);
                
                \draw[dashed] (v1) -- node[weight, midway, right=5pt] {$1$} (v3);
        \end{tikzpicture}
                
        \vspace{2mm}
        (a) The graph $G$
        \end{minipage}
        \hfill
        \begin{minipage}{0.47\textwidth}
        \centering
        \begin{tikzpicture}[
                scale=1.05,
                transform shape,
                vertex/.style={
                circle,
                draw,
                fill=white,
                inner sep=1.5pt,
                minimum size=9mm
                },
                weight/.style={
                draw=none,
                fill=none,
                inner sep=0pt
                },
                every path/.style={line width=0.9pt}
        ]
                \node[vertex] (v1) at (90:2.2)   {$v_1$};
                \node[vertex] (v2) at (18:2.2)   {$v_2$};
                \node[vertex] (v3) at (-54:2.2)  {$v_3$};
                \node[vertex] (v4) at (-126:2.2) {$v_4$};
                \node[vertex] (v5) at (162:2.2)  {$v_5$};
                
                \draw[dashed, gray] (v1) -- (v2);
                
                \draw (v2) -- node[weight, midway, right=4pt] {$2$} (v3);
                \draw (v3) -- node[weight, midway, below=4pt] {$2$} (v4);
                \draw (v4) -- node[weight, midway, left=4pt] {$2$} (v5);
                \draw (v5) -- node[weight, midway, above left=2pt] {$2$} (v1);
                \draw (v1) -- node[weight, midway, right=5pt] {$1$} (v3);
        \end{tikzpicture}
                
        \vspace{2mm}
        (b) The cactus $H=G-v_1v_2$
        \end{minipage}
                
        \caption{An illustration of the weighted example on five vertices.}
        \label{fig:no-cut-vertex}
\end{figure}


\begin{thebibliography}{99}

\bibitem{akhtarphilip2026}
Sheikh Shakil Akhtar and Geevarghese Philip.
\newblock Exact algorithms for edge deletion to cactus.
\newblock In Florent Foucaud and Aline Parreau, editors,
{\em Combinatorial Algorithms}, IWOCA 2026, Lecture Notes in Computer Science,
vol.~16587, pages 32--44. Springer, Cham, 2026.
\newblock \doi{10.1007/978-3-032-27732-9_3}.

\bibitem{bellman1962}
Richard Bellman.
\newblock Dynamic programming treatment of the travelling salesman problem.
\newblock {\em Journal of the ACM}, 9(1):61--63, 1962.
\newblock \doi{10.1145/321105.321111}.

\bibitem{bjorklund2007}
Andreas Bj\"orklund, Thore Husfeldt, Petteri Kaski, and Mikko Koivisto.
\newblock Fourier meets M\"obius: Fast subset convolution.
\newblock In {\em Proceedings of the 39th Annual ACM Symposium on Theory of Computing
(STOC)}, pages 67--74. ACM, 2007.
\newblock \doi{10.1145/1250790.1250801}.

\bibitem{cai1996graph}
Leizhen Cai.
\newblock Fixed-parameter tractability of graph modification problems for hereditary properties.
\newblock {\em Information Processing Letters}, 58(4):171--176, 1996.
\newblock \doi{10.1016/0020-0190(96)00050-6}.

\bibitem{elmallahcolbourn1988}
E.~S. El-Mallah and C.~J. Colbourn.
\newblock The complexity of some edge deletion problems.
\newblock {\em IEEE Transactions on Circuits and Systems}, 35(3):354--362, 1988.
\newblock \doi{10.1109/31.1748}.

\bibitem{fominkratsch2010}
Fedor~V. Fomin and Dieter Kratsch.
\newblock {\em Exact Exponential Algorithms}.
\newblock Texts in Theoretical Computer Science. An EATCS Series. Springer,
Berlin, Heidelberg, 2010.
\newblock \doi{10.1007/978-3-642-16533-7}.

\bibitem{golumbicjamison1985b}
Martin Charles Golumbic and Robert~E. Jamison.
\newblock Edge and vertex intersection of paths in a tree.
\newblock {\em Discrete Mathematics}, 55(2):151--159, 1985.
\newblock \doi{10.1016/0012-365X(85)90043-3}.

\bibitem{golumbicjamison1985a}
Martin Charles Golumbic and Robert~E. Jamison.
\newblock The edge intersection graphs of paths in a tree.
\newblock {\em Journal of Combinatorial Theory, Series B}, 38(1):8--22, 1985.
\newblock \doi{10.1016/0095-8956(85)90088-7}.

\bibitem{hararyuhlenbeck1953}
Frank Harary and George~E. Uhlenbeck.
\newblock On the number of Husimi trees, I.
\newblock {\em Proceedings of the National Academy of Sciences of the United States of America},
39(4):315--322, 1953.
\newblock \doi{10.1073/pnas.39.4.315}.

\bibitem{heldkarp1962}
Michael Held and Richard~M. Karp.
\newblock A dynamic programming approach to sequencing problems.
\newblock {\em Journal of the Society for Industrial and Applied Mathematics},
10(1):196--210, 1962.
\newblock \doi{10.1137/0110015}.

\bibitem{hopcrofttarjan1973}
John Hopcroft and Robert Tarjan.
\newblock Algorithm 447: Efficient algorithms for graph manipulation.
\newblock {\em Communications of the ACM}, 16(6):372--378, 1973.
\newblock \doi{10.1145/362248.362272}.

\bibitem{kochpardaldossantos2024}
Ivo Koch, Nina Pardal, and Vinicius Fernandes dos Santos.
\newblock Edge deletion to tree-like graph classes.
\newblock {\em Discrete Applied Mathematics}, 348:122--131, 2024.
\newblock \doi{10.1016/j.dam.2024.01.028}.

\bibitem{natanzonshamirsharan2001}
Assaf Natanzon, Ron Shamir, and Roded Sharan.
\newblock Complexity classification of some edge modification problems.
\newblock {\em Discrete Applied Mathematics}, 113(1):109--128, 2001.
\newblock \doi{10.1016/S0166-218X(00)00391-7}.

\bibitem{west2001}
Douglas~B. West.
\newblock {\em Introduction to Graph Theory}.
\newblock 2nd edition. Prentice Hall, Upper Saddle River, NJ, 2001.

\bibitem{white2001graphs}
A.~T. White.
\newblock {\em Graphs of Groups on Surfaces: Interactions and Models}.
\newblock North-Holland Mathematics Studies, vol.~188. Elsevier, Amsterdam, 2001.

\bibitem{yannakakis1981}
Mihalis Yannakakis.
\newblock Edge-deletion problems.
\newblock {\em SIAM Journal on Computing}, 10(2):297--309, 1981.
\newblock \doi{10.1137/0210021}.

\end{thebibliography}
\end{document}